\documentclass[journal,draftclsnofoot,onecolumn]{IEEEtran}
\usepackage{amsmath}
\usepackage{amssymb}
\usepackage{latexsym}
\newcommand{\C}{{\mathcal{C}}}

\newcommand{\E}{\mathbb{E}}

\newcommand{\tr}{{\rm tr}}

\newcommand{\0}{\boldsymbol{0}}

\newtheorem{proposition}{Proposition}

\newtheorem{lemma}{Lemma}
\newtheorem{corollary}{Corollary}
\newtheorem{theorem}{Theorem}

\title{Contract Theory Approach to\\ Incentivizing Market and Control Design}

\author{Yasuaki Wasa, Kenji Hirata and Kenko Uchida
\thanks{Yasuaki Wasa (corresponding author) and Kenko Uchida
are with the Department of Electrical Engineering and Bioscience,
Waseda University, Tokyo 169-8555, JAPAN.
Kenji Hirata is with the Graduate School of Science and Engineering for Research, 
University of Toyama, Toyama 930-8555, JAPAN.
{\tt\small wasa@aoni.waseda.jp, hirata@eng.u-toyama.ac.jp, kuchida@waseda.jp}.
}}

\begin{document}
\maketitle

\begin{abstract}
We discuss an incentivizing market and model-based approach 
to design the energy management and control systems 
which realize high-quality ancillary services in dynamic power grids. 
Under the electricity liberalization, 
such incentivizing market should secure a high speed market-clearing 
by using the market players' private information well. 
Inspired by contract theory in microeconomics field, 
we propose a novel design method of such incentivizing market 
based on the integration of the economics models and the dynamic grid model. 
The conventional contract problems are analyzed for 
static systems or dynamical systems with control inputs directly operated by the principal. 
The analysis is, however, in discord with the incentivizing market. 
The main challenge of our approach is to reformulate the contract problems 
adapted to the market from the system and control perspective. 
We first establish the fundamental formulas for optimal design, 
and clarify the basic properties of the designed market. 
We also discuss possibilities, limitation and some challenges 
in the direction of our approach and general market-based approaches.
\end{abstract}
%\begin{IEEEkeywords}
%\noindent 
%Cyber-Physical systems,
%Environmental monitoring,
%Visual sensor networks,
%Payoff-based learning,
%Game theoretic cooperative control
%\end{IEEEkeywords}

%------------------------------------------------------------------------%
% Main Document
\IEEEpeerreviewmaketitle

\section{Introduction}
\label{sec:1}

Achieving a quality assurance of electric energy, called the ancillary service, 
is a key target of next-generation energy management and control systems 
for dynamic electric smart grids 
where electricity liberalization is fully enforced and 
renewable energy is highly penetrated \cite{bib01}. 
Frequency, voltage and power controls, 
which are typical contents of the ancillary service, 
have been technical requirements for the electric energy supplier (e.g., see \cite{bib02,bib03}). 
Since the electricity liberalization starts, 
such ancillary control services have been investigated and realized 
in competitive electricity markets \cite{bib04,bib05,bib06}. 
In view these, future energy management and control systems 
should include ancillary service markets with some incentive mechanisms, as core elements, 
which provide high-quality and fast-response control services to the extent of the primary level. 
Moreover, if we need ancillary control services of transient state, 
ancillary service markets should include physical models of dynamic power grids. 
In this article, we propose an incentivizing market-based approach 
to design the energy management and control systems 
which realize high-quality ancillary services in such dynamic power grids. 
Using this approach, we develop a design method of such incentivizing market 
based on the integration of the economics models and the dynamic grid model, 
and provide fundamental conditions and formulas for the incentivizing market design.

Our approach is developed under the assumption 
that an energy dispatch scheduling on a future time interval 
has been finished in a spot energy market 
at the tertiary control level \cite{bib07,bib08}, e.g., for one hour future interval, 
and that each agent has a linearized model of his/her own system 
along the scheduled trajectory over the future time interval. 
For this linear time-varying model, 
we formulate a design problem of energy management and control systems 
based on a real-time regulation market, called the ancillary market, 
at the secondary and primary control levels \cite{bib07,bib08}. 
Participants in the dynamic electric smart grid are consumers, suppliers or prosumers, called agents, 
who control their physical system selfishly according to their own criterion, 
and utility (independent public commission), 
who integrates economically all the controls of agents into a high-quality power demand and supply. 
In the integration, a market mechanism is adopted inevitably 
in order to secure selfish behaviors of agents in electricity liberalization; 
that is, each agent bids his/her certain private information 
in response to a market-clearing price, 
while utility (auctioneer) clears the market based on the bidding and decides the prices, in real-time.

The market model in our approach is characterized by two terminologies: 
\emph{private information} and \emph{incentivizing market}. 
An iterative market-clearing model so-called the t\^atonnement model 
does not need rigorous models, 
but does not generally guarantee the convergence to a specified equilibrium. 
Moreover, if it converges, the t\^atonnement model takes a long time 
to converge at a market clearing equilibrium. 
To overcome the issues, we propose a novel model-based and market-based approach 
that designs first some incentives for the agents 
to report their private information (including their own model information) 
to the utility in the market, and makes it possible to realize a high speed market-clearing. 
This approach needs incentivizing costs, 
and the resulting optimization process can be recognized 
as an intermediate model (the second best model) between two extremal models, 
namely the t\^atonnement model 
and the so-called supply/demand function equilibrium model (the first best model) 
which uses for free all agents' rigorous models, i.e., agents' private information. 
We provide this approach with fundamental formulas and tools 
to design the incentivizing mechanism in the market, 
and discuss the basic properties of the designed market. 
We also discuss the relationships of our incentivizing mechanism 
with the Lagrange multiplier based integration/decomposition mechanism and the mechanism design. 
%We also discuss possibilities and limitation of our approach together with some other challenges.

This article has been organized as follows: 
Section~\ref{sec:2} introduces a dynamic power grid model 
and a model-based incentivizing market model. 
We next derive some theoretical results on a general reward design problem in Section~\ref{sec:3}. 
In Section~\ref{sec:4}, we show the relationship between the private information 
and the incentives and discuss possibilities and limitation of our approach 
through three typical scenarios. 
In Section~\ref{sec:5}, we summarize the results.

\section{Grid Model and Incentivizing Market Model}
\label{sec:2}

%\subsection{Two Layers Market}
%\label{sec:2.1}
%

In this paper, we consider the two level architecture with the two layers market; 
spot energy market and real-time regulation market. 
The well-known temporally-separated architecture \cite{bib07,bib08} 
motivated by the conventional power system control 
is divided into the primary control level (voltage and frequency stabilization), 
the secondary control level (quasi-stationary power imbalance control) 
and the tertiary control level (economic dispatch). 
The two layers market reorganizes the conventional three-level architecture 
according to the functions of the markets. 
Our approach is developed under the assumption 
that an energy dispatch scheduling on a future time interval has been finished 
in a spot energy market (at the tertiary control level), 
and that each agent has a linearized model of his/her own system 
along the scheduled trajectory over the future time interval. 
For this linear time-varying model, 
we formulate a design problem of energy management and control systems 
to realize ancillary services based on a real-time regulation market 
(at the secondary and primary control levels).

\subsection{Linearized Grid Model}
Let us first consider the linearized time-varying model used in the ancillary market. 
This paper considers one of the standard grid models, 
the average system frequency model \cite{bib09}, 
as a generic model of high speed response for ancillary service control problems 
with two area power networks and with two kinds of players: \emph{Utility} and \emph{Agents}. 
Here we present a linearized model of each player's own system 
along the scheduled trajectory over a future time interval 
during when an energy dispatch scheduling has been finished 
in a spot energy market.

The utility dynamics, which describes the deviation of 
the power and/or frequency balance and other deviations 
from physical constraints as well, obeys the following equation: 
\begin{eqnarray}
dx_{0t} 
= 
\left( A_{00}(t) x_{0t} + A_{01}(t) x_{1t} + A_{02}(t) x_{2t} \right) dt 
+ D_0(t) d\beta_t, \quad 
t_0 \leq t \leq t_f, 
\label{eq:u_dyn}
\end{eqnarray}
and is evaluated by the utility's revenue functional:
\begin{equation}
J_0(t,x;u) = \E_{t,x} 
\left[ \varphi_0(t_f, x_{t_f} ) + \int_{t}^{t_f} l_0( \tau, x_\tau, u_\tau ) d\tau \right]
\end{equation}
where 
$x=(x_0^\top, x_1^\top, x_2^\top)^\top\in\mathbb{R}^n$ is 
the collection of the states of 
the utility dynamics $x_0\in\mathbb{R}^{n_0}$ 
and the agents' dynamics $x_i\in\mathbb{R}^{n_i}$, $i=1, 2$, 
at time $t\in[t_0,t_f]$, 
and 
$u=( u_1^\top, u_2^\top )^\top$ is the local control inputs, respectively; 
$\beta_t$ is the disturbance modeled by a standard Wiener process on $[t_0,t_f]$; 
$\E_{t,x}$ indicates an expectation given initial data $(t,x)$; 
we use an abbreviation like $x_{0t}=x_0(t)$, $x_t=x(t)$. 
The dynamics of the agent $i$ ($i=1,2$) obeys %the following equation:
\begin{equation}
dx_{it} = ( A_i(t) x_{it} + B_i(t) u_{it} )dt + D_i(t) d\beta_t, 
\quad t_0\leq t\leq t_f, \ i=1,2, 
\label{eq:a_dyn}
\end{equation}
and is evaluated by the agent's revenue functional:
\begin{equation}
J_i(t,x;u) = \E_{t,x} 
\left[ \varphi_i(t_f, x_{t_f} ) + \int_{t}^{t_f} l_i( \tau, x_\tau, u_\tau ) d\tau \right], \quad i=1,2. 
\end{equation}
The agent's state $x_i$ indicates typically 
the deviation of power generation or consumption from the scheduled trajectory; 
the control $u_i$ compensates the deviation. 
An admissible control of agent $i$, denoted as $u_i\in\Gamma_i$, 
is a state feedback $u_{it}=u_i(t,x)$ denoted by 
$u_i:[t_0,t_f]\times \mathbb{R}^n \to U_i \subset \mathbb{R}^{m_i}$ 
is continuous at $t\in[t_0,t_f]$ and Lipschitz continuous at $x\in\mathbb{R}^n$, 
$i=1,2$. 
To simplify the description in the following, 
let us describe the grid dynamics by combining the utility dunamics (\ref{eq:u_dyn}) 
and the agent's dynamics (\ref{eq:a_dyn}) as follows: 
\begin{equation}
d x_t 
= f(t,x_t,u_t) dt + D(t) d\beta_t 
:= \left( f_0(t,x_t) + f_1(t,x_{1t},u_{1t}) + f_2(t,x_{2t},u_{2t}) \right) dt
+ D(t) d\beta_t
\end{equation}
where
\begin{equation*}
f_0 = 
\left( \begin{array}{c}
A_{00} x_0 + A_{01} x_1 + A_{02} x_2 \\
\0 \\
\0
\end{array} \right), 
\ 
f_1 = 
\left( \begin{array}{c}
\0 \\
A_{1} x_1 + B_{1} u_1 \\
\0 
\end{array} \right), \ 
f_2 = 
\left( \begin{array}{c}
\0 \\
\0 \\
A_{2} x_2 + B_{2} u_2
\end{array} \right), \ 
D = 
\left( \begin{array}{c}
D_0 \\
D_1 \\
D_2 
\end{array} \right). 
\end{equation*}

We need the following assumptions, 
which make the discussions in this paper mathematically regorous. 
The notations 
$\nabla_t=\partial/\partial t$, 
$\nabla_x=(\partial/\partial x_0, \partial/\partial x_1, \partial/\partial x_2)$ 
and 
$\nabla_x^2=[ \partial^2/\partial x_i \partial x_j ]$ 
are used. 
\begin{description}
\item[(A1)]~Each element of matrices 
$A_{00}(t)$, $A_{0i}(t)$, $A_i(t)$, $B_i(t)$, $D_0(t)$, $D_i(t)$, $i=1,2$, 
are continuous at $t\in[t_0,t_f]$, and $D(t)D(t)^\top > 0$ for all $t\in[t_0,t_f]$.
\item[(A2)]~The set $U_i$, $i=1,2$, are compact and convex. 
\item[(A3)]~The function $\varphi_0(t_f,\cdot):\mathbb{R}^n\to \mathbb{R}$ is 
of class $\C^2$ and $\nabla_x \varphi_0(t_f,\cdot)$ is polynomial growth. 
The function $l_0:[t_0,t_f]\times \mathbb{R}^n\times U_1 \times U_2 \to \mathbb{R}$ is 
$\C^1$ at $(t,u_1,u_2)\in[t_0,t_f]\times U_1 \times U_2$ 
and $\C^2$ at $x\in\mathbb{R}^n$, 
and $\nabla_x l_0$, $\nabla_{u_i} l_0$, $i=1,2$, 
are polynomial growth at $(x,u_1,u_2)\in\mathbb{R}^n\times U_1 \times U_2$. 
\item[(A4)]~The function $\varphi_i(t_f,\cdot):\mathbb{R}^{n_i}\to \mathbb{R}$ is 
of class $\C^2$ and $\nabla_x \varphi_i(t_f,\cdot)$ is polynomial growth. 
The function $l_i:[t_0,t_f]\times \mathbb{R}^{n_i}\times U_i \to \mathbb{R}$ is 
$\C^1$ at $t\in[t_0,t_f]$ 
and $\C^2$ at $(x_i,u_i)\in\mathbb{R}^{n_i}\times U_i$, 
and $\nabla_{x_i} l_i$, $\nabla_{u_i} l_i$ 
are polynomial growth at $(x_i,u_i)\in\mathbb{R}^{n_i}\times U_i$ 
and $\nabla_{u_i}^2 l_i < 0$, $i=1,2$. 
\end{description}

We formulated the grid model with the evaluation functionals 
on the finite time interval $[t_0,t_f]$. 
For simplicity, from now on, 
we consider the state feedback strategies $u=( u_1^\top, u_2^\top )^\top$ 
derived by dynamic programming. 
We will discuss possibilities of the other options in Section~\ref{sec:4}. 
To achieve the objective, 
we reformulate our problems on the future time interval from the current time $t$ 
to the final time $t_f$ based on the time-consistency property.

To describe formulas concisely, 
we adopt the continuous-time model in this article; 
we can develop in parallel the same results in the discrete-time model. 
On the other hand, 
to develop our discussion in the continuous-time model in a mathematically sound way, 
we need some technical assumptions as stated above 
and in the later discussion; 
however, the assumptions except that on convexity (or concavity) 
are for assuring an appropriate smoothness and boundness of the variables appearing 
in the discussions, 
but not essential for developing our key ideas.

\subsection{Incentivizing Market Model}

To describe market mechanism, we need to specify participant's private information. 
Private information of agent $i=1,2$ consists of model information 
$\Xi_i=(f_i,\varphi_i,l_i)$ and on-line information $Z_{it} \subset \{ x_{it}^{t_f}, u_{it}^{t_f} \}$, 
where 
$x_{it}^{t_f} := \{ x_{i\tau}, t\leq \tau \leq t_f \}$ and 
$u_{it}^{t_f} := \{ u_{i\tau}, t\leq \tau \leq t_f \}$.

To incentivize agent's behavior in market model, 
we (or a market planner) use a reward (salary) functional of the following form. 
The reward (salary) functional:
\begin{eqnarray}
W_i^w(t,x_t^{t_f};u)
= 
w_{if}(t_f,x_{t_f}) + w_{i0}(t,x) 
+ \int_{t}^{t_f} w_{i1}(\tau,x_\tau) d\tau
+ \int_{t}^{t_f} w_{i2}(\tau,x_\tau) dx_\tau, 
\quad
t_0 \leq t \leq t_f, \ i=1,2, 
\end{eqnarray}
are defined along with the grid dynamics
$ d x_\tau = f(\tau,x_\tau,u_\tau) d\tau + D(\tau) d\beta_\tau$
where 
$w=(w_1,w_2)$ and $w_i=(w_{if},w_{i0},w_{i1},w_{i2})$. 
Admissible parameters of the reward functional, 
denoted as $w=(w_1,w_2)\in\Pi\times \Pi$, 
are defiend by:
$w_{if}(t_f,\cdot): \mathbb{R}^n\to \mathbb{R}$ is of class $\C^2$ 
and $\nabla_x w_{if}(t_f,\cdot)$ is polynomial growth; 
$w_{i0}:[t_0,t_f]\times\mathbb{R}^n\to \mathbb{R}$ is continuous 
at $(t,x)\in[t_0,t_f]\times\mathbb{R}^n$; 
$w_{i1}:[t_0,t_f]\times\mathbb{R}^n\to \mathbb{R}$ 
are $\C^1$ at $t\in[t_0,t_f]$ and $\C^2$ at $x\in\mathbb{R}^n$, 
and $\nabla_x w_{i1}$ is polynomial growth at $x\in\mathbb{R}^n$; 
$w_{i2}:[t_0,t_f]\times\mathbb{R}^n\to \mathbb{R}^{1\times n}$ is of class $\C^1$ 
and polynomial growth at $x\in\mathbb{R}^n$, and $\nabla_x w_{i2}$ is bounded. 
We use the notation $W_i^w$ so as to emphasize 
the dependence of $W_i$ on the choice of the parameter $w$. 
In the following discussion, 
we often use the same notation to show such parameter dependence. 
We try to express the parameter $w$ with another parameter $h$, 
which we call the price, 
so that the reward functional depends on the choice of the price $h$; 
then such dependence is also denoted as $W_i^h$.

The reward functionals together with the utility's revenue functional 
and the agent's revenue functional define the social welfare functional as:
\begin{equation}
I^w(t,x;u) = J_0(t,x;u) 
- \E_{t,x} \left[ W_1^w(t,x_t^{t_f};u) \right]
- \E_{t,x} \left[ W_2^w(t,x_t^{t_f};u) \right]
\end{equation}
and the agent's profit functional as: 
\begin{eqnarray}
I_i^w(t,x;u) = J_i(t,x;u) 
+ \E_{t,x} \left[ W_i^w(t,x_t^{t_f};u) \right], 
\quad
i=1,2.
\end{eqnarray}

A market planner designs a market mechanism with incentivizing structures 
and makes auction rules as well, 
based on the evaluation functionals and the grid model information 
introduced so far; 
the auction is performed in the following five steps: 
\begin{description}
\item[Step 1:]~Utility announces the auction system, and agents decide participation.
\item[Step 2:]~Agent offers his/her bid based on his/her own private information.
\item[Step 3:]~Based on agents' bids, price is determined so as to maximize social welfare. 
\item[Step 4:]~Agent decides his/her control to maximize his/her own profit based on price.
\item[Step 5:]~Utility pay rewards to agents.
\end{description}
Note that Steps 2, 3 and 4 will be performed continuously over a finite time interval.

\section{Model-based One-shot Market Mechanism}
\label{sec:3}

\subsection{Reward Design for Incentivizing}
\label{sec:3.1}

Components of our market model and their general interplay have been described in the previous section. 
To complete our market model, we need to fix a concrete shape of agents' bidding, 
and design reward functionals by choosing their characterizing parameter, 
called the reward parameter, $w=(w_1,w_2)$. 
First, let us specify agents' private information to be bidden 
in the market model discussed here: 
\emph{Each agent's model information $\Xi_i=(f_i,\varphi_i,l_i)$ is sent a priori to utility, 
and each agent's on-line information to be bidden is just 
the current state, i.e., $Z_{it}=x_{it}$, 
which means that utility cannot access control input $u_i$.} 
Then, the design problem of our market is reduced to a social welfare maximization problem, 
called the reward design problem, 
subject to the constraints that provide the market with two incentivizing functions by rewards, 
which is formulated as follows:
\begin{eqnarray}
&& 
\max_{u\in\Gamma_1\times\Gamma_2,\ w\in\Pi\times\Pi} 
I^w(t,x;u)
\nonumber \\
& \mbox{subject to} &
\nonumber \\
& \mbox{(Constraint 1)} & 
I^w_1(t,x;u) = \max_{v_1\in\Gamma_1} I_1^w(t,x;v_1,u_2), \quad
I^w_2(t,x;u) = \max_{v_2\in\Gamma_2} I_2^w(t,x;u_1,v_2), 
\nonumber \\
& \mbox{(Constraint 2)} &
I^w_1(t,x;u) \geq k_1(t,x), \qquad\qquad\qquad\
I^w_2(t,x;u) \geq k_2(t,x), 
\nonumber 
\end{eqnarray}
where $k_i:[t_0,t_f]\times\mathbb{R}^n\to \mathbb{R}$ is 
continuous at $(t,x)\in[t_0,t_f]\times\mathbb{R}^n$. 
By solving this problem, 
we obtain the optimal reward functional with two incentive functions 
and the agents' optimal controls. 
Constraint 1 claims that 
the reward incentivizes each agent's behavior to adopt the optimal control 
that maximizes her own profit and, 
in other words, constitutes a Nash equilibrium together with the other agent's control. 
This also implies that, since the utility holds the bidden models, 
the utility can know the control profile, 
even if it is not bidden. 
On the other hand, Constraint 2 assures a prescribed level of each agent's profit. 
The above formulation is an application of the moral hazard problem 
in contract theory \cite{bib10,bib11} to our market design problem; 
using terminology of contract theory, 
we call Constraint 1 and Constraint 2 
the incentive compatibility constraint and the individual rationality constraint, respectively. 
The conventional contract (moral hazard) problems analyzed 
for static systems and dynamical systems with control inputs directly 
operated by the principal \cite{bib10}. 
The analysis is, however, in discord with the incentivizing market. 
The main challenge of our incentivizing market design is 
to reformulate the moral problems adapted to the market as above 
and synthesize the proposed market from the system and control perspective.

\subsection{Solutions for General Reward Design}
\label{sec:3.2}

To solve the reward design problem, 
we start specifying a form of the reward functionals by using Constraints 1 and 2. 
For a parameter
$w=(w_1,w_2)\in\Pi\times\Pi$, 
let $(u_1^w,u_2^w)$ be a pair of optimal controls 
(a Nash equilibrium in $\Gamma_1\times\Gamma_2$) 
defined by $u_i^w=\arg\max_{u_i\in\Gamma_i} I_i^w(t,x;u_i,u_{-i}^w)$, 
$i=1,2$, 
so that Constraint 1 is fulfilled, 
where $u_{-1}:=u_2$ and $u_{-2}:=u_1$.
Then, as shown in Appendix, 
the Hamilton-Jacobi-Bellman (HJB) equations for the value functions 
\begin{equation}
V_i^w(t,x) = \max_{u_i\in\Gamma_i} I_i^w(t,x;u_i,u_{-i}^w) - w_{i0}(t,x), 
\quad i=1,2,
\label{eq:agent_Vfunc}
\end{equation}
lead the reward functional to the form:
\begin{eqnarray}
\!\!\!\!\!\!\!\!\!\!\!\! && 
W_i^w(t,x_t^{t_f};u_i,u_{-i}^w)
= h_{i0}^w(t,x) - \varphi_i(t_f,x_{it_f})
\nonumber \\
\!\!\!\!\!\!\!\!\!\!\!\! && \qquad
- \int_{t}^{t_f} \left[ h_{i1}^w(\tau,x_t)f(\tau,x_\tau,u_1^w(\tau,x_\tau),u_2^w(\tau,x_\tau)) 
+ l_i(\tau,x_{it},u_i^w(\tau,x_\tau)) \right] d\tau
+ \int_{t}^{t_f} h_{i1}^w(\tau,x_{\tau}) dx_{\tau}
\label{eq:V-W}
\end{eqnarray}
along with 
$dx_{\tau} = f(\tau,x_\tau,u_i(\tau,x_\tau),u_{-i}^w(\tau,x_\tau))d\tau+D(\tau)d\beta_\tau$, 
$t\leq \tau \leq t_f$, 
where $h_{i0}^w$ and $h_{i1}^w$ are defined by
\begin{subequations}
\begin{eqnarray}
h_{i0}^w(t,x) &=& V_i^w(t,x) + w_{i0}(t,x) 
\label{eq:V-h0}
\\
h_{i1}^w(t,x) &=& \nabla_x V_i^w(t,x) + w_{i2}(t,x)
\label{eq:V-h1}
\end{eqnarray}
\label{eq:V-h}
\end{subequations}
Moreover, $u_i^w$, $i=1,2$, which constitute a Nash equilibirum, 
must satisfy
\begin{eqnarray}
u_i^w(\tau,x) 
&=& 
\arg\max_{u_i\in U_i} 
\left[ h_{i1}^w(\tau,x) f(\tau,x,u_i,u_{-i}^w(\tau,x)) + l_i(\tau,x_i,u_i) + w_{i1}(\tau,x) \right]
\nonumber \\
&=& 
\arg\max_{u_i\in U_i} 
\left[ h_{i1}^w(\tau,x) f_i(\tau,x_i,u_i) + l_i(\tau,x_i,u_i) \right]
\nonumber 
\end{eqnarray}
so that a function $\mu_i$ given in Lemma~\ref{lem:01} 
provides uniquely $u_i^w$ with the expression of an explicit dependence on $h_{i1}^w$ 
such that 
$u_i^w(\tau,x) = \mu_i( \tau, x_i, h_{i1}^w(\tau,x) )$. 
For simplicity of notation, 
we will denote sometimes $\mu_i(\tau,x_i,h_{i1}^w(\tau,x))$ by $\mu_i^{h_{i1}^w}(\tau,x)$.

\begin{lemma}
\label{lem:01}
There exists a unique function $\mu_i$ that satisfies 
\begin{eqnarray}
\mu_i(\tau,x_i,p_i) = 
\arg\max_{u_i\in U_i} \left[ p_i f_i(\tau,x_i,u_i) + l_i(\tau,x_i,u_i) \right], 
\quad i=1,2, 
\end{eqnarray}
for each $(\tau,x_i,p_i) \in [t_0,t_f] \times \mathbb{R}^{n_i} \times \mathbb{R}^{1 \times n}$ 
such that $\mu_i$ is continuous at $(\tau,x_i,p_i)$ and Lipschitz continuous at $(x_i,p_i)$. 
\end{lemma}
\begin{proof}
The continuity at $(\tau,x_i,p_i)$ follows from the uniqueness of the maximum. 
The Lipschitz continuity is shown by Lemma VI.6.3 in \cite{bib12}. 
\end{proof}

Now, summarizing the above observation, 
we see that, in solving the reward design problem, 
Costraint 1 enables us to limit a search of the optimal reward functional 
to the class of the form (\ref{eq:V-W}). 
In this form of the reward functional, 
$h_{i0}^w$ and $h_{i1}^w$ are given by (\ref{eq:V-h0}) and (\ref{eq:V-h1}), respectively, 
which implies that they depend on a choice of the parameter 
$w\in(w_1,w_2)\in\Pi\times\Pi$. 
We can show that this class of reward functionals is invariant, 
even if the class of parameters $h_i^w=(h_{i0}^w,h_{i1}^w)$ is generalized 
to a class where dependence on the parameter $w$ is not necessarily required. 
For this purpose, 
let $h_i=(h_{i0},h_{i1})$ and define a class of reward parameters 
$h=(h_1,h_2)\in H \times H$ such that 
$h_{i0}:[t_0,t_f]\times\mathbb{R}^n\to\mathbb{R}$ 
is continuous at $(t,x)\in[t_0,t_f]\times\mathbb{R}^n$; 
$h_{i1}:[t_0,t_f]\times\mathbb{R}^n\to\mathbb{R}^{1\times n}$ 
is of class $\C^1$ and polynomial growth at $x\in\mathbb{R}^n$, 
and $\nabla_x h_{i1}$ is bounded. 
Note that $h^w=(h_1^w,h_2^w)\in H\times H$ 
for any $w=(w_1,w_2)\in\Pi\times\Pi$ 
if $(V_i^w,\nabla_x V_i^w)$ is in the class $H$.

\begin{proposition} 
\label{prop:01}
(a) 
A pair of controls $(u_1^w,u_2^w)$ constitutes 
a Nash equilibrium satisfying Constraint 1 
for a pair of reward functionals $(W_1^w,W_2^w)$ 
with a parameter $w=(w_1,w_2)\in\Pi\times\Pi$ 
and the corresponding pair of value functions $(V_1^w,V_2^w)$ satisfies 
the condition that $(V_i^w,\nabla_x V_i^w)$, $i=1,2$ 
are in the class $H$, 
only if there is a parameter $h=(h_1,h_2)\in H \times H$ 
such that the pair of reward functionals has the form 
\begin{eqnarray}
\!\!\!\!\!\!\!\!\!\!\!\! &&
W_i^h(t,x_t^{t_f};u_i,\mu_{-i}^{h_{-i1}})
= h_{i0}^w(t,x) - \varphi_i(t_f,x_{it_f})
\nonumber \\
\!\!\!\!\!\!\!\!\!\!\!\! && \qquad
- \int_{t}^{t_f} \left[ h_{i1}(\tau,x_t)f(\tau,x_\tau,\mu_1^{h_{11}}(\tau,x_\tau),\mu_{2}^{h_{21}}(\tau,x_\tau)) 
+ l_i(\tau,x_{it},\mu_i^{h_{i1}}(\tau,x_\tau)) \right] d\tau
+ \int_{t}^{t_f} h_{i1}(\tau,x_{\tau}) dx_{\tau}
\label{eq:h-W}
\end{eqnarray}
along with 
$ dx_\tau = f(\tau,x_\tau,u_i(\tau,x_\tau),\mu_{-i}^{h_{-i1}}(\tau,x_\tau))d\tau + D(\tau)d\beta_\tau$, 
$i=1,2$. 
\\
(b) 
For the reward functionals (\ref{eq:h-W}) 
with a parameter $h=(h_1,h_2)\in H\times H$, 
a pair of controls $(u_1,u_2)\in\Gamma_1\times\Gamma_2$ 
is a Nash equilibrium if and only if it has the form 
\begin{equation}
u_i(\tau,x) = \mu_i( \tau, x_i, h_{i1}(\tau,x) ), \quad i=1,2. 
\label{eq:Nash_mu}
\end{equation}
(c)
For the reward functionals (\ref{eq:h-W}) 
with a parameter $h=(h_1,h_2)\in H\times H$ 
and the Nash equilibrium (\ref{eq:Nash_mu}), 
Constraint 2 is fulfilled if and only if 
$h_{i0}$, $i=1,2$, are specified such as 
$h_{i0}(t,x) \geq k_{i}(t,x)$. 
\end{proposition}
\begin{proof}
(a) 
We have already seen that, 
for a chosen parameter $w=(w_1,w_2)\in\Pi\times\Pi$, 
the reward functionals for which $(u_1^w,u_2^w)$ 
constitutes a Nash equilibrium must have the form (\ref{eq:V-W}) 
with the parameters (\ref{eq:V-h0}) and (\ref{eq:V-h1}), 
and the Nash equilibrium must be given as 
$u_i^w(\tau,x)=\mu_i(\tau,x_i,h_{i1}^w(\tau,x))$, $i=1,2$. 
Now, let $h=(h_1,h_2)\in H\times H$ 
be chosen independently of $w$ and set a reward parameter 
$\bar{w}=(\bar{w}_1,\bar{w}_2)\in\Pi\times\Pi$ as
\begin{subequations}
\begin{eqnarray*}
\bar{w}_{if}(t_f,x) &=& 
- \varphi_i(t_f,x) 
\\
\bar{w}_{i0}(t,x) &=&
h_{i0}(t,x)
\\
\bar{w}_{i1}(t,x) &=& 
- h_{i1}(\tau,x) f(\tau,x, \mu_1^{h_{11}}(\tau,x),\mu_{2}^{h_{21}}(\tau,x)) - l_i(\tau,x_i,\mu_i^{h_{i1}}(\tau,x))
\\
\bar{w}_{i2} &=& 
h_{i1}(\tau,x).
\end{eqnarray*}
\end{subequations}
Then, we can show that, 
for these reward parameters, 
the HJB equation (\ref{eq:DP_V}) in Appendix has a unique constant solution of the form 
$V_i^{\bar{w}}(\tau,x)=0$, so that we have 
$h_{i0}^{\bar{w}}(t,x) = V_i^{\bar{w}}(t,x) + \bar{w}_{i0}(t,x) = h_{i0}(t,x)$ 
and 
$h_{i1}^{\bar{w}}(\tau,x) = \nabla_x V_i^{\bar{w}}(\tau,x) + \bar{w}_{i2}(\tau,x) = h_{i1}(\tau,x)$. 
This implies that the class of reward functionals given by 
(\ref{eq:V-W}) with (\ref{eq:V-h0}) and (\ref{eq:V-h1}) is invariant, 
even if the class of parameters $h^w=(h_1^w,h_2^w)$ depending on $w$ 
is generalized to $H\times H$, 
and proves the part (a) of this proposition. 
\\
(b) 
For reward functionals of the form (\ref{eq:h-W}) 
with a parameter $h=(h_1,h_2)\in H \times H$, 
profit functionals of the agent $i=1,2$ are represented as
{\arraycolsep=2pt
\begin{eqnarray}
I_i^h(t,x_t^{t_f};u_i,\mu_{-i}^{h_{-i1}}) 
&=& h_{i0}(t,x) 
- \int_t^{t_f} 
\left[ h_{i1}(\tau,x_\tau) f(\tau,x_\tau, \mu_1^{h_{11}}(\tau,x_\tau), \mu_{2}^{h_{21}}(\tau,x_\tau)) 
+ l_i(\tau,x_{i\tau},\mu_i^{h_{i1}}(\tau,x_\tau)) \right] d\tau
\nonumber \\
&&
+ \int_t^{t_f} 
\left[ h_{i1}(\tau,x_\tau) f(\tau,x_\tau, u_i(\tau,x_\tau), \mu_{-i}^{h_{-i1}}(\tau,x_\tau)) 
+ l_i(\tau,x_{i\tau},u_i(\tau,x_\tau)) \right]
d\tau.
\label{eq:I_h}
\end{eqnarray}}
From the definition of $\mu_i$ given in Lemma~\ref{lem:01}, 
the second (integral) term in the right hand side of the identity above is non-negative, 
and therefore the pair of controls $(u_1,u_2)\in\Gamma_1\times\Gamma_2$ 
is a Nash equilibrium if and only if 
$u_i=\mu_i^{h_{i1}}$, $i=1,2$. 
\\
(c) 
It is obvious because the identity (\ref{eq:I_h}) 
guarantees $I_i^h(t,x_t^{t_f};\mu_1^{h_{11}},\mu_{2}^{h_{21}})=h_{i0}(t,x)$, $i=1,2$. 
\end{proof}

A key message of part (a) in Proposition~\ref{prop:01} is that 
the original parameter $w=(w_1,w_2)$ can be replaced 
with the parameter $h=(h_1,h_2)$. 
We will see below that the parameter $h$ can be interpreted as a price (vector), 
and show that it enables us to introduce a dynamic contract, 
which realizes requisite incentives, in the market model. 
Another message from the parts (a) and (b) is that 
we can shift the Nash equilibrium (\ref{eq:Nash_mu}) freely to some extent 
by selecting the price (vector) $h$. 
The parameterization of reward (salary) functional with the parameter $h$ 
would be in itself a new result of interest in contract theory, 
which is different from the known types based on typically 
the so-called first order condition \cite{bib11,bib13,bib14} 
and the other types \cite{bib15,bib16} in the contract theory literatures. 
Finally, note that for proving this proposition 
we do not use the linearity of the grid model in the state, 
while we need the linearity and additivity in the controls in the grid model 
and the convexity (concavity) of the control ranges 
and the revenue functions in (A2) and (A4) as well.

Now, using Proposition~\ref{prop:01}, 
we can present an optimal control based approach, 
in which the parameter $h=(h_1,h_2)\in H \times H$ 
plays a role of control, to the reward design problem.

\begin{theorem}
\label{thm:01}
The reward design problem with the parameter $h=(h_1,h_2)\in H \times H$ 
is equivalent to an optimal control problem described by 
\begin{equation*}
\max_{(h_1,h_2)\in H\times H} 
I^h(t,x;\mu_1^{h_{11}},\mu_2^{h_{21}})
\end{equation*}
subject to 
\begin{equation*}
h_{i0}(t,x) \geq k_{i}(t,x), \quad i=1,2, 
\end{equation*}
and the stochastic state equation:
\begin{equation*}
dx_\tau = 
f(\tau,x_\tau, \mu_1^{h_{11}}(\tau,x_\tau),\mu_2^{h_{21}}(\tau,x_\tau))d\tau + D(\tau) d\beta_\tau,
\quad t \leq \tau \leq t_f,
\end{equation*}
where $\mu_i^{h_{i1}}(\tau,x)=\mu_i(\tau,x_i,h_{i1}(\tau,x))$, $i=1,2$. 
\end{theorem}
\begin{proof}
From (a) of Proposition~\ref{prop:01} that 
the social welfare functional $I^h(t,x;\mu_1^{h_{11}},\mu_2^{h_{21}})$ 
is represented by 
\begin{eqnarray}
I^h(t,x;\mu_1^{h_{11}},\mu_2^{h_{21}}) 
&=& 
J_0(t,x;\mu_1^{h_{11}},\mu_2^{h_{21}}) 
- \E_{t,x} \left[ \sum_{i=1}^2 W_i^h( t, x_t^{t_f} ; \mu_1^{h_{11}},\mu_2^{h_{21}} ) \right]
\nonumber \\
&=& \E_{t,x} \left[ 
\varphi_0(t_f,x_{t_f}) + \int_t^{t_f} l_0(\tau,x_\tau,\mu_1^{h_{11}}(\tau,x_{\tau}),\mu_2^{h_{21}}(\tau,x_{\tau})) d\tau 
\right]
\nonumber \\
&&
+ \E_{t,x} \left[
\sum_{i=1}^2 \left( \varphi_i(t_f,x_{it_f}) + \int_t^{t_f} l_i(\tau,x_{i\tau},\mu_i^{h_{i1}}(\tau,x_\tau)) d\tau \right)
\right] - \sum_{i=1}^2 h_{i0}(t,x). 
\label{eq:I0_h}
\end{eqnarray}
Then, 
from (b) and (c) of Proposition~\ref{prop:01}, 
Constraints 1 and 2 are fulfilled, respectively, 
for any $h=(h_1,h_2)\in H\times H$. 
Thus we have this theorem. 
\end{proof}

The optimal solution $h^*=(h_1^*,h_2^*)\in H \times H$ 
leads to the Nash equilibrium $(\mu_1^{h_{11}^*},\mu_2^{h_{21}^*})$, $i=1,2$. 
Note that $h_{i0}^*(t,x)=k_{i}(t,x)$, $i=1,2$ 
follows from the expression (\ref{eq:I0_h}), and then Constraint 2 is fulfilled.

\section{Discussion through Typical Scenarios}
\label{sec:4}

It is generally difficult to solve the optimal control problem in Theorem~\ref{thm:01}. 
Here, focusing on some special cases, we discuss qualitative properties of the parameter $h$ 
and try to give economic meanings to the parameter and the reward functional.

Let the value function be denoted by 
\begin{equation}
V(t,x) = 
\sup_{(h_1,h_2)\in H\times\times H} 
I^h(t,x;\mu_1^{h_{11}},\mu_2^{h_{21}}) + h_{10}(t,x) + h_{20}(t,x).
\end{equation}
Then, the HJB equation is given by
\begin{subequations}
\arraycolsep=2pt
\begin{eqnarray}
\nabla_t V(t,x) &+& \frac{1}{2} \left[ \nabla_x^2 V(t,x) D(t) D(t)^\top \right]
\nonumber \\
&+& 
\sup_{(h_{11},h_{21}) \in \mathbb{R}^{1\times n}\times\mathbb{R}^{1\times n}} 
\Big[
\nabla_x V(t,x) f_0(t,x) 
+ l_0(t,x_i,\mu_1(t,x_1,h_{11}),\mu_2(t,x_2,h_{21})) 
\nonumber \\
&& \qquad
+ \sum_{i=1}^2 \left( 
\nabla_x V(t,x) f_i(t,x_i,\mu_i(t,x_i,h_{i1})) 
+ l_i(t,x_i,\mu_i(t,x_i,h_{i1}))
\right) \Big] 
= 0. 
\\
V(t_f,x_{t_f}) &=& \varphi_0(t_f,x_{t_f}) + \varphi_1(t_f,x_{1t_f}) + \varphi_2(t_f,x_{2t_f}).
\end{eqnarray}
\label{eq:DP_principal}
\end{subequations}
In this section, 
we discuss the relationship between the private information and the incentives, 
interpretation and limitation of our approach through three typical cases.

{\bf (A)} 
Consider the case when the utility evaluates only the grid state $x$ 
and does not evaluate the agents' control inputs $u_i$, $i=1,2$ 
such that $l_0=l_0(\tau,x)$.

\begin{corollary}
\label{coro:01}
In Case (A), 
if the HJB equation (\ref{eq:DP_principal}) 
has a solution $V(t,x)$ such that $(V,\nabla_x V)$ is in the class $H$, 
the optimal parameters $h_i^*=(h_{i0}^*,h_{i1}^*)\in H$, $i=1,2$ 
are given by $h_{i0}^*(t,x) = k_i(t,x)$ and $h_{i1}^*(t,x)=\nabla_x V(t,x)$, 
$t_0 \leq t \leq t_f$. 
\end{corollary}
\begin{proof}
$h_{i0}^*(t,x) = k_i(t,x)$ is already noted. 
In Case (A), 
the maximization in the HJB equation (\ref{eq:DP_principal}) becomes 
\begin{eqnarray*}
\sum_{i=1}^2 \sup_{h_{i1}\in\mathbb{R}^{1\times n}} 
\left[ \nabla_x V(t,x) f_i(t,x_i,\mu_i(t,x_i,h_{i1})) 
+ l_i(t,x,\mu_i(t,x_i,h_{i1})) \right], 
\end{eqnarray*}
and it follows from Lemma~\ref{lem:01} that 
the maximum is attained by $h_{i1}^*(t,x)=\nabla_x V(t,x)$, $i=1,2$. 
Then, the verification theorem \cite[Theorem VI.4.1]{bib11} 
verifies the optimality of the parameter. 
\end{proof}

The fact $h_{i1}^*(t,x) = \nabla_x V(t,x)$, $i=1,2$, 
shown in Corollary~\ref{coro:01} implies that 
the reward parameter $h_{11}^*(t,x) (=h_{21}^*(t,x))$ 
can be regarded as a price of quantity $x$ at time $t$; 
$\nabla_x V(t,x)$ is actually called the shadow price in economics literatures, 
and our parametrization of the reward functional could be suitable for the market model. 
Note that the form of the utility's revenue function as $l_0=l_0(\tau,x)$ is no so restrictive, 
since the utility dynamics has no control input.

{\bf (B)} 
Consider the case that 
the utility's revenue functional is given by 
\begin{eqnarray*}
J_0(t,x;u) = 
\E_{t,x} \left[ 
\varphi_0(t_f,x_{t_f}) + \int_{t}^{t_f} l_0(\tau,x_\tau) d\tau
+ \sum_{i=1}^2 \left( \varphi_i(t_f,x_{it_f}) + \int_t^{t_f} l_i(\tau,x_{i\tau},u_{i\tau}) d\tau \right) \right].
\end{eqnarray*}
That is, 
the utility's revenue is the sum of the original utility's revenue 
and the agent's revenues. 
Assume further that 
the payment of the rewards for the agents 
is not liquidated in the social welfare, 
i.e., the utility's revenue functional above is identical to 
the social welfare functional, and the agents ask no profit, 
i.e., $k_i(t,x) \equiv 0$, $i=1,2$. 
In this case, 
as the problem is basically equivalent to 
that in Case (A) with $k_i(t,x) \equiv 0$, $i=1,2$, 
repeating the same argument as in Case (A), 
we can obtain the same result as 
Corollary~\ref{coro:01} with $h_{i0}^*(t,x) \equiv 0$, $i=1,2$. 
The result shows that, 
if the price vector $\nabla_x V(t,x)$, 
called the adjoint vector in the optimal control theory, 
is provided by the utility, 
each agent can realize his/her optimal control in a decentralized way 
such as $\mu_i^{h_{i1}}(\tau,x) = \mu_i(\tau,x_i,\nabla_x V(\tau,x))$, $i=1,2$; 
this result corresponds to the dual decomposition of 
the static optimization based on ``Lagrange multiplier" (price). 
On the other hand, 
from the viewpoint of the incentive design, 
each agent in Case (B)
has a zero level of incentive to the participation in the market 
(the decentralized optimization based on the price), 
because he/she obtains no profit, 
whereas, in Case (A), 
agents have the profits $k_i(t,x)$, $i=1,2$, rewarded by the utility 
and have the incentives to the participation. 
We see that 
the implementation of this decentralized optimization scheme 
may require additional incentives or legal forces for strategic agents.

{\bf (C)}
The reward design discussed so far incentivizes agents 
to constitute a Nash equilibrium (Constraint 1)
and to participate in the market 
if the profit level is over his/her expectation (Constraint 2). 
However, 
these are assured under the tacit assumption that 
the agents' private information consisting of the model data and the on-line data 
is truthfully sent and bidden; 
if an agent fictitiously bids his/her private information, 
for example, 
the Nash equilibrium shifts or disappears; 
the \emph{mechanism design} \cite{bib17,bib18} 
provides a solution in such case by using additional incentives. 
Consider the same setting as in Case (B) where 
the social welfare functional given as above and does not include 
the budget for payment of the agents' rewards, 
and, on the other hand, 
let the agent's profit functional have an additional reward functional as 
\begin{eqnarray*}
I_i^w(t,x;u) = J_i(t,x;u) 
+ \E_{t,x} \left[ W_i^w(t,x_t^{t_f};u) \right]
+ \E_{t,x} \left[ W_{ai}^w(t,x_t^{t_f};u) \right]
\end{eqnarray*}
where 
\begin{eqnarray*}
\E_{t,x}\left[ W_{a1}^w(t,x_t^{t_f};u) \right] 
= J_0(t,x;u) + J_{2}(t,x;u), &&
\E_{t,x}\left[ W_{a2}^w(t,x_t^{t_f};u) \right] 
= J_0(t,x;u) + J_{1}(t,x;u)
\end{eqnarray*}
and $l_0=l_0(\tau,x)$. 
In this case, 
replacing 
$\varphi_i(t_f,x_{it_f})$ with 
$\varphi_0(t_f,x_{t_f}) + \sum_{i=1}^2 \varphi_i(t_f,x_{it_f})$ 
and also 
$l_i(\tau,x_i,u_i)$ with 
$l_0(\tau,x) + \sum_{i=1}^2 l_i(\tau,x_i,u_i)$ 
and repeating the same argument as in Case (A), 
we have the same conclusion as in Corollary~\ref{coro:01}, 
so that 
the agents should constitute a Nash equilibrium and participate in the market. 
In this case, moreover, 
the agents should report his/her model information 
and bid his/her on-line information truthfully; 
the reason for this is as follows. 
First, note that the additional reward $W^{ai}$ provides 
the utility and all the agents with the same revenue, 
so that the optimal price from the viewpoint of the social welfare 
is optimal for all the agents. 
Second, note that the utility calculates the optimal price based on 
the reported model and the bidden states. 
Therefore, if an agent sends or bids fictitiously his/her private information to the market, 
the agent obtains a price which is not optimal for his/her own profit. 
This incentivizing scheme corresponds to the Groves mechanism \cite{bib17} 
in mechanism design literatures. 
Finally, we point out an issue of this scheme; 
the rewards $W^{ai}$ should be additionally paid from the social welfare budget.

\section{Conclusions}
\label{sec:5}

On the basis of a genetic model suggested from the average system frequency model \cite{bib09}, 
we have discussed the incentivizing market and model-based approach 
to design the energy management and control systems 
which realize ancillary services in dynamic power grids. 
The key issue of the approach is to incentivize the agents (areas) to open their private information, 
which is essential to realize our model-based scheme, to the utility. 
We have proposed a design method of such incentivizing market 
by integrating the economics models and tools with the dynamic physical model, 
and clarified its basic properties of use together with its possibilities and limits for further developments.

%\appendices
\appendix

Based on the principle of optimality, 
the value function (\ref{eq:agent_Vfunc}) 
leads to the HJB equation:
\begin{subequations}
\arraycolsep=2pt
\begin{eqnarray}
\nabla_t V_i^w(t,x) &+& \frac{1}{2} \tr\left[ \nabla_x^2 V_i^w(t,x) D(t)D(t)^\top \right]
\nonumber \\
&=& 
- \max_{u_i\in U_i} \left[ \left( \nabla_x V_i^w(t,x) + w_{i2}(t,x) \right) 
f(t,x,u_i,u_{-i}^w(t,x)) + l_i(t,x_i,u_i) + w_{i1}(t,x) \right]
\nonumber \\
&=& 
- \left[ \left( \nabla_x V_i^w(t,x) + w_{i2}(t,x) \right) 
f(t,x,u_i^w(t,x),u_{-i}^w(t,x)) + l_i(t,x_i,u_i^w(t,x)) + w_{i1}(t,x) \right], 
\\
V_i^w(t_f,x_{t_f}) &=& \varphi_i(t_f,x_{it_f}) + w_{if}(t_f,x_{t_f}). 
\end{eqnarray}
\label{eq:DP_V}
\end{subequations}
Substituting the above relation into the right hand side of the Ito's differential equality 
\begin{equation*}
dV_i^w(t,x_t) 
= 
\left( \nabla_t V_i^w(t,x_t) 
+ \frac{1}{2} \tr\left[ \nabla_x^2 V_i^w(t,x_t) D(t)D(t)^\top \right] \right) dt
+ \nabla_x V_i^w(t,x_t) dx_t
\end{equation*}
along with 
$dx_t = f(t,x_t,u_i(t,x_t),u_{-i}^w(t,x_t))dt + D(t)d\beta_t$, 
and integrating the both side on $[t,t_f]$, 
we have the reward functional of 
the form (\ref{eq:V-W}) with (\ref{eq:V-h0}) and (\ref{eq:V-h1}).


\nocite{*}
\begin{thebibliography}{99}
\bibitem{bib01}
M. Amin, A.M. Annaswamy, C.L. DeMarco and T. Samad, 
{\it IEEE vision for smart grid controls: 2030 and beyond}, 
IEEE Press, 2013.

\bibitem{bib02}
M.D. Ilic and S.X. Liu, 
{\it Hierarchical Power Systems Control -- Its Value in a Changing Industry}, 
Springer, 1996.

\bibitem{bib03}
Y.G. Rebours, D.S. Kirschen, M. Trotignon and S. Rossignol, 
``A survey of frequency and voltage control ancillary services -- Part I: Technical features," 
{\it IEEE Trans. Power Systems}, vol. 22, no. 1, pp.350--357, 2007.

\bibitem{bib04}
A.J. Wood and B.F. Wollenberg, 
{\it Power Generator Operation and Control}, 
Wiley, 1996.

\bibitem{bib05}
M.A.B. Zammit, D.J. Hill and R.J. Kaye, 
``Designing ancillary services markets for power system security," 
{\it IEEE Trans. Power Systems}, vol. 15, no. 2, pp. 675--680, 2000.

\bibitem{bib06}
E. Ela, V. Gevorgian, A. Tuohy, B. Kirby, M. Milligan and M. O'Malley, 
``Market designs for the primary frequency response ancillary service -- Part I: motivation and design," 
{\it IEEE Trans. Power Systems}, vol. 29, no. 1, pp. 421--431, 2014.

\bibitem{bib07}
M.D. Ilic, 
``Toward a unified modeling and control for sustainable and resilient electric energy systems," 
{\it Foundations and Trends in Electric Energy Systems}, vol. 1, no. 1--2, pp. 1--141, 2016.

\bibitem{bib08}
A. Kiani, A. Annaswamy and T. Samad, 
``A hierachical transactive control architecture for renewables integration in smart grids: Analytical modeling and stability," 
{\it IEEE Trans. Smart Grid}, vol. 5, no. 4, pp. 2054--2065, 2014.

\bibitem{bib09}
A.W. Berger and F.C. Schweppe, 
``Real time pricing to assist in load frequency control," 
{\it IEEE Trans. Power Systems}, vol. 4, no. 3, pp. 920--926, 1989.

\bibitem{bib10}
P. Bolton and M. Dewatripont, 
{\it Contract Theory}, 
The MIT Press, 2005.

\bibitem{bib11}
B. Holmstrom and P. Milgrom, 
``Aggregation and linearity in the provision of intertemporal incentives," 
{\it Econometrica}, vol. 55, no. 2, pp. 303--328, 1987.

\bibitem{bib12}
W.H. Fleming and R.W. Rishel, 
{\it Deterministic and Stochastic Optimal Control}, 
Springer, 1975.

\bibitem{bib13}
H. Schattler and J. Sung, 
``The first-order approach to the continuous time principal-agent problem with exponential utility," 
{\it J. Economic Theory}, vol. 61, pp. 331--371, 1993.

\bibitem{bib14}
H.K. Koo, G. Shim and J. Sung, 
``Optimal multi-agent performance measures for team contracts," 
{it Mathematical Finance}, vol. 18, no. 4, pp. 649--667, 2008.

\bibitem{bib15}
Y. Sannikov, 
``Contracts: The theory of dynamic principal-agent relationships and the continuous-time approach," 
In: D. Acemoglu, M. Arellano and E. Dekel (Eds.), 
{\it Advances in Economics and Econometrics, 10th World Congress of the Econometric Society}, 
Cambridge University Press, 2013.

\bibitem{bib16}
J. Cvitanic and J. Zhang, {\it Contract Theory in Continuous-Time Models}, 
Springer, 2013

\bibitem{bib17}
M.O. Jackson, 
``Mechanism theory," 
In: U. Derigs (Ed.), {\it Encyclopedia of Life Support Systems}, 
EOLSS Publishers, 2003. 

\bibitem{bib18}
Y. Okajima, T. Murao, K. Hirata, and K. Uchida, 
``A dynamic mechanism for LQG power networks with random type parameters and pricing delay," 
{\it Proc. 52nd IEEE Conf. Decision and Control}, pp. 2384--2390, 2013.


\bibitem{MB17}
J. Moon and T. Basar, 
``Linear quadratic risk-sensitive and robust mean field games," 
{\it IEEE Trans. Automatic Control}, vol. 62, no. 3, pp. 1062--1077, 2017.


\end{thebibliography}
\end{document}